\documentclass[11pt]{article}
\usepackage[margin=1.5in]{geometry}

\usepackage{lineno,hyperref}

\usepackage{tikz}
\modulolinenumbers[5]

\usepackage{amssymb,amsmath,amsthm}
\usepackage{graphicx}
\usepackage[labelformat=simple,labelsep=period]{subfig}

\usepackage{authblk}

\newtheorem{theorem}{Theorem}

\newtheorem{definition}{Definition}

\newcommand{\evcn}{\alpha^{\infty}}
\newcommand{\mvc}{\alpha}

\title{(Eternal) Vertex Cover Number of Infinite and Finite Grid Graphs}

\author{Tiziana Calamoneri}
\affil{Università di Roma La Sapienza, Roma, Italy\\
\texttt{calamo@di.uniroma1.it}}

\author{Federico Cor\`{o}}
\affil{University of Perugia, Italy\\
\texttt{federico.coro@unipg.it}}

\date{}

\begin{document}

\maketitle

\begin{abstract}
In the eternal vertex cover problem, mobile guards on the vertices of a graph are used to defend it against an infinite sequence of attacks on its edges by moving to neighbor vertices. The eternal vertex cover problem consists in determining the minimum number of necessary guards. 
Motivated by previous literature, in this paper, we study the vertex cover and eternal vertex cover problems on regular grids,  when passing from infinite to finite version of the same graphs, and we provide either coinciding or very tight lower and upper bounds on the number of necessary guards. 
To this aim, we generalize the notions of minimum vertex covers and minimum eternal vertex cover in order to be well defined for infinite grids.
\end{abstract}


\section{Introduction}

The eternal vertex cover problem can be described in terms of a two-player game and models a problem associated with defending the vertices of a given graph $G$ against a sequence of attacks:
at the beginning of the game, the defender (which controls some guards lying on vertices of $G$) chooses a placement of guards on a subset of vertices of \(G\), defining an initial configuration.
Subsequently, in each round of the game, the attacker attacks an edge of their choice. To repel an attack, a guard from an incident vertex moves across the attacked edge to defend it; at the same time, other guards may either remain where they are or move to neighbor vertices.
If the defender is able to do this, then the attack has been successfully defended and the game proceeds to another round with the attacker choosing the next edge to attack. 
Otherwise, the attacker wins. 

Graph protection has its historical roots in the time of the ancient Roman Empire~\cite{klostermeyer2016protecting} and currently finds application in network security and redundant file storage handling.

If the defender is able to keep defending against any infinite sequence of attacks on \(G\) with \(k\) guards, then we say that there is a {\em defense strategy on \(G\) with \(k\) guards}. 
This strategy requires that the set of vertices containing guards is a vertex cover before and after each round. 
In this case, the set of positions of all the guards in any round of the game defines a configuration, referred to as an {\em eternal vertex cover of \(G\) of size \(k\)}.

The {\em eternal vertex cover problem} on a graph $G$ requires to find the minimum value of $k$ for which a defense strategy on \(G\) with \(k\) guards exists. Such minimum is denoted by $\evcn (G)$.

\bigskip

{\em Literature.}
The problem was first formulated in 2009 by Klostermeyer and Mynhardt~\cite{klostermeyer2009edge} and, in the same paper, it was proved that, for any graph $G$, \(\mvc(G ) \le \evcn(G ) \le 2\mvc(G )\), where $\mvc(G)$ is the cardinality of a minimum vertex cover for $G$.

The problem  has been deeply studied  from a computational complexity point of view:
deciding whether \(k\) guards can protect all the edges of a graph is NP-hard~\cite{fomin2010parameterized} 
and remains NP-hard even for internally triangulated planar graphs~\cite{babu2019new}; moreover, it is APX-hard and there is a 2-approximation algorithm; finally, exact (exponential) algorithms are given.

Nevertheless, there are a number of special graphs for which the problem can be exactly solved: 
trees~\cite{klostermeyer2009edge}
(\(\evcn(T)=n- |L(T )| + 1\), where \(L(T )\) is the number of leaves of \(T\)),
cacti~\cite{babu2020linear},
simple generalizations of trees~\cite{araki2015eternal}
({\em i.e.} graphs constructed by replacing each edge of a tree by an arbitrary elementary bipartite graph or by an arbitrary clique),
cycles on $n$ vertices~\cite{klostermeyer2009edge}
(\(\evcn(C_n)=\mvc(C_n)=\lceil\frac{n}{2} \rceil\)),
chordal graphs~\cite{babu2019new}
and $n \times m$ grids~\cite{klostermeyer2009edge}
(whose $\evcn$ is \(\frac{nm}{2} = \mvc\) if \(nm\) is even and is \(\lceil\frac{nm}{2}\rceil = \mvc+1\) if \(n,m>1\) are odd with \(n\ge m\)).
More in general, in~\cite{klostermeyer2016protecting,klostermeyer2009edge}, some conditions to have \(\evcn = \mvc\) are provided and it seems particularly interesting to understand in which cases these two parameters are very close.

Finally, there are a number of papers ({\em e.g.}~\cite{babu2019new,anderson2012mortal,anderson2014graphs,klostermeyer2011graphs}) connecting $\mvc$ and $\evcn$ of special graphs with other parameters, such that the (eternal) domination number, the vertex connectivity number, and the minimum cardinality of a vertex cover that contains all cut vertices.

\bigskip

{\em Our Work.}
This study arises from observing that, while any eternal vertex cover of an infinite path must puts guards alternately on its vertices, for a finite $n$ vertex path $P_n$, $\evcn(P_n)=n-1$~\cite{klostermeyer2009edge}.

On the other hand, it is immediate to see that any eternal vertex cover of the infinite squared grid puts guards alternately on its vertices; this is true also in the finite case, indeed,
as we have already stated, 
from the literature we know that a squared grid
({\em i.e.} the Cartesian product of two paths) has $\evcn$ roughly equal to the half of its number of vertices.

So, not all the graphs behave in the same way when passing from the infinite to the finite version.
We wonder which is the behavior of some naturally infinite graphs, that is whether we get an increase of their eternal vertex cover number when we reduce to a finite portion.

To this aim, we first generalize the notions of minimum vertex cover and minimum eternal vertex cover in order to make them well defined for infinite grids.
Then we consider infinite regular grid graphs.
For each of them, we evaluate on which portion of vertices it is necessary to put a guard;
moreover we consider a finite (rectangular) portion of such graphs and study the value of their $\evcn$. 
Finally we compare our results in the infinite and finite case.
The results of this paper are summarized in Tables~\ref{tab:infinite} and~\ref{tab:finite}.

It is worth noting that, to the best of our knowledge, in the literature, no generalizations are known for the notions of minimum vertex cover and eternal vertex cover of infinite graphs.

\begin{table}[ht]
\centering
\begin{tabular}{|l||c|c|}
\hline
Infinite graphs & $\rho$          & $\rho^{\infty}$ \\
\hline \hline
Path     &    $=1/2$              &$=1/2$    \\
Squared grid & $=1/2$ (Th.~\ref{th.vc_infiniteT3})   & $=1/2$  (Th.~\ref{th.infinite_EVC_T3_T4_T6})              \\
Hexagonal grid  & $=1/2$ (Th.~\ref{th.vc_infiniteT3})          & $=1/2$ (Th.~\ref{th.infinite_EVC_T3_T4_T6})                  \\
Triangular grid & 
$= 2/3$ (Th.~\ref{th.vc_infiniteT6}) & $= 2/3$ (Th.~\ref{th.infinite_EVC_T3_T4_T6})       \\
Octagonal grid & 
$= 3/4$ (Th.~\ref{th.vc_infiniteT8}) &  $= 3/4$ (Th.~\ref{th.infinite_EVC_T3_T4_T6})       \\
\hline
\end{tabular}
\vspace*{.3cm}
\caption{Summary of the results for infinite grid graphs; the lower part of the table contains the results proved in this paper.}
\label{tab:infinite}
\end{table}

\begin{table}[ht]
\centering
\hspace*{-1cm}
\begin{tabular}{|l||c|c|}
\hline
Finite graphs  & $\rho$          & $\rho^{\infty}$\\
\hline \hline
Path $P_n$ &
$\frac{1}{2}(1-\frac{1}{n}) \leq \rho \leq \frac{1}{2}$~\cite{klostermeyer2009edge} &
=$1-\frac{1}{n}$~\cite{klostermeyer2009edge}     \\
Squared grid 
& $\frac{1}{2} -\frac{1}{2hw} \leq \rho \leq \frac{1}{2}$ 
& 
$ \frac{1}{2} -\frac{1}{2hw} \le \rho \le\frac{1}{2} +\frac{1}{2hw}$~\cite{klostermeyer2009edge} \\
\hline
Hexagonal grid  & $= \frac{1}{2}$ (Th.~\ref{th.finite_vc_T3_T6}) & $=\frac{1}{2}$ (Th.~\ref{th.finite_EVC_T3})              \\
Triangular grid & $\frac{2}{3}-\frac{1}{3w} \leq \rho \leq \frac{2}{3}$ (Th.~\ref{th.finite_vc_T3_T6})          & $\frac{2}{3}-\frac{1}{3w} \leq \rho^{\infty} \leq \frac{2}{3}+\frac{4}{h}+\frac{4}{h^2}$ (Th.~\ref{th.finite_EVC_T6})     \\
Octagonal grid  & $\frac{3}{4} -\frac{1}{4w} \leq \rho \leq \frac{3}{4}$ (Th.~\ref{th.finite_vc_T3_T6}) & $\frac{3}{4}  -\frac{1}{4w}\leq \rho^\infty \leq \frac{3}{4} + \frac{1}{2h} + \frac{1}{4h^2}$ (Th.~\ref{th.finite_EVC_T8})  \\
\hline
\end{tabular}
\vspace*{.3cm}
\caption{Summary of the results for finite grid graphs; the lower part of the table contains results proved in this paper. The results of last two rows related to $\rho^\infty$ holds for $h \ge w$.
}
\label{tab:finite}
\end{table}

\section{Definitions}

In this paper we consider first {\em infinite regular grids}, {\em i.e.} infinite plane graphs in which all faces are identical regular polygons: hexagonal, squared and triangular grids, here denoted by $T_{\Delta}$, where $\Delta$ represents the degree that is equal to 3,4,6, respectively.
Although they are inspired by regular tassellations of the plane, we do not need the measure of edges and angles so we
keep only the underlying graph structure; for completeness, we study also the regular grid of degree 8, called octagonal grid and denoted by $T_8$. 
Moreover, 
in order to make the exposition easier, we embed them on the Euclidean plane and force the vertices of these grids to have integer coordinates, so for every pair of coordinates $(x,y)$, $x,y \in \mathbb{Z}$, there is a vertex.
It follows a formal description of these 4 grids:

\smallskip

    \noindent{\bf hexagonal grid $T_3$} (see Figure~\ref{fig:infinite-grid-hex}): every vertex $(x,y)$ is connected both with $(x, y+1)$ and with $(x, y-1)$; moreover, for every \(x,k \in \mathbb{Z}\), every vertex \((x, 4k)\) is connected with \((x+1, 4k+1)\), every vertex \((x, 4k+1)\) is connected with \((x-1, 4k)\), every vertex \((x, 4k+2)\) is connected with \((x-1, 4k+3)\) and, finally, every vertex \((x, 4k+3)\) is connected with \((x+1, 4k+2)\);

    \noindent {\bf squared grid $T_4$} (see Figure~\ref{fig:infinite-grid-square}): every vertex $(x,y)$ is connected with $(x+1, y)$, $(x-1, y)$, $(x, y+1)$ and $(x, y-1)$;

    \noindent {\bf triangular grid $T_6$} (see Figure~\ref{fig:infinite-grid-tri}): every vertex $(x,y)$ is connected with $(x+1, y)$, $(x-1, y)$, $(x, y+1)$, $(x, y-1)$, $(x-1, y-1)$ and $(x+1, y+1)$;

    \noindent {\bf octagonal grid.} (see Figure~\ref{fig:infinite-grid-octagonal}): every vertex $(x,y)$ is connected with $(x+1, y)$, $(x-1, y)$, $(x, y+1)$, $(x, y-1)$, $(x-1, y-1)$, $(x+1, y+1)$, $(x-1, y+1)$ and $(x+1, y-1)$.

\begin{figure}[ht]
\centering
\subfloat[hexagonal grid\label{fig:infinite-grid-hex}]{
\includegraphics[scale=0.8]{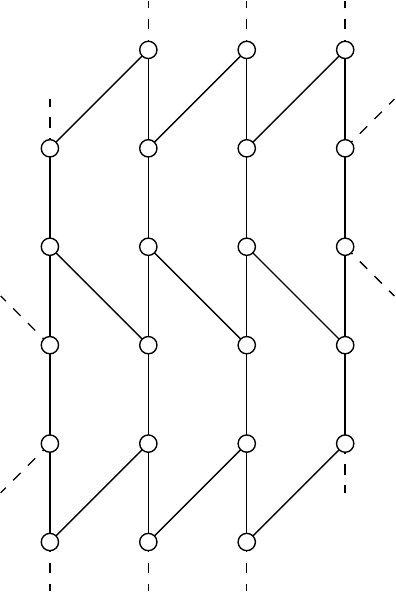}
}
~
\subfloat[squared grid\label{fig:infinite-grid-square}]{
\includegraphics[scale=0.8]{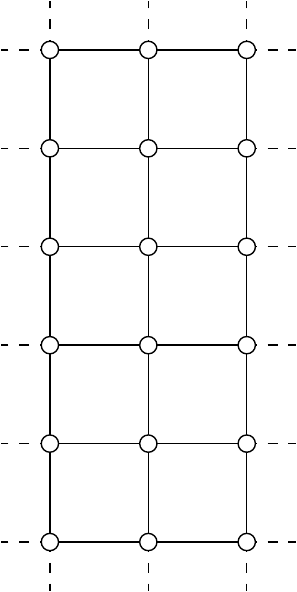}
}~
\subfloat[triangular grid\label{fig:infinite-grid-tri}]{
\includegraphics[scale=0.8]{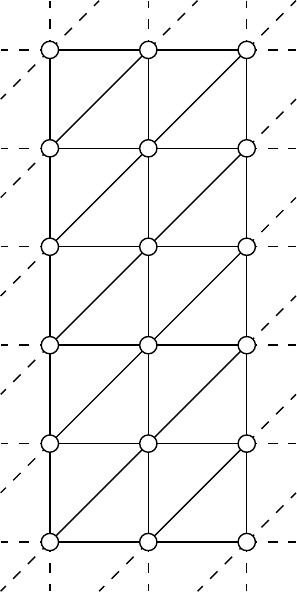}
}~
\subfloat[octagonal grid\label{fig:infinite-grid-octagonal}]{
\includegraphics[scale=0.8]{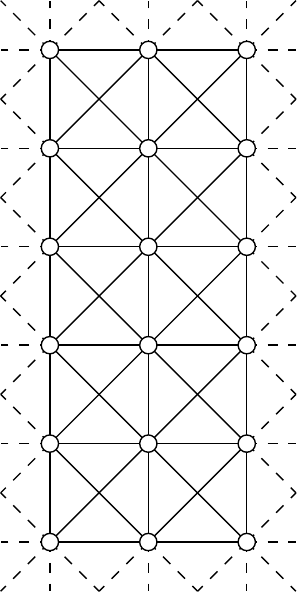}
}
\caption{The four considered infinite regular grid graphs.}
\label{fig:grids}
\end{figure}

Then we study also finite portions of infinite grid graphs, that are induced by the vertices of a set of regular faces contained into a $h \times w$ rectangle with sides parallel to the axes and vertices at integer coordinates, $h, w \in \mathbb{N}$; we will call these subgraphs as {\em finite rectangular grids} and denote them with $T_{\Delta}(h,w)$; 
in the following, we will omit the term ``rectangular'' for brevity.
Clearly, this concept is not well defined for octagonal grids since they are not planar graphs, but it is not difficult to extend it (for example by momentarily removing bottom-right to up-left edges, determining the subgraph according to the definition, and then adding again the removed edges).
In view of their definition, finite grids do not contain degree 1 vertices.
Moreover, we assume the finite grids do not degenerate in simpler structures (and hence, {\em e.g.} the cases $h=1$ and $w=1$ are not allowed in any grid, while the case $h < 4$ is not allowed in the hexagonal grid).

\medskip

We now give a formal definition of the problems we want to study on infinite and finite grids.
It worth noting that, while the following definitions are very well known  in the case of finite graphs, we need to extend them in order to make them working on infinite grids.

\begin{definition}
Given a (either finite or infinite) graph $G=(V,E)$, a {\em vertex cover} for $G$ is a set of vertices $C \subseteq V$ that include at least one endpoint of every edge. 
If $G$ is a {\em finite} graph, the {\em vertex cover problem} consists in finding a vertex cover of minimum cardinality, and this number is denoted with $\mvc (G)$.
If $G$ is an {\em infinite} grid, the {\em vertex cover problem} consists in finding a vertex cover $C$ such that, there exists $n_0 >0$ and, for every $n \geq n_0$ there is a finite rectangular grid of $G$, $G(n,n)$ with vertex set $V(n)$, such that $C \cap V(n)$ is a minimum vertex cover for $G(n,n)$.

In both cases, we will say that a solution of the vertex cover problem is a {\em minimum vertex cover}.
\end{definition}

Given a vertex cover $S$ for $G$, an {\em attack} may occur on a single edge $e = \{u, v\}$, where $u \in S$; a {\em defense} by $S$ to the attack on the edge is a one-to-one function $f : S \rightarrow V$ such that:\\
(1) $f(u) = v$, and \\
(2) for each $s \in S \setminus \{ u \}$, $f(s) \in N[s]$ where $N[s]$ is the close neighborhood of $s$.

Given any vertex $u \in S$, we say that {\em the guard on $u$ shifts to $f(u)$} and, by extension, $S$ {\em shifts to} $S'$ where $S'=\{ f(s) \,\,\mbox{ s.t. } s \in S\}$. 

Since an attack of an edge whose both its extremes contain a guard can always be repelled without changing the configuration of guards (simply swapping the position of the guards on that edge), we only consider attacks on edges with one unguarded vertex.

\begin{definition}
Given a (either finite or infinite) graph $G=(V,E)$, a vertex cover $S$ for $G$ is an {\em eternal vertex cover} if every (possibly infinite) sequence of attacks can be defended, that is if a defense shifts $S$ in $S'$ and $S'$ is an eternal vertex cover. 
If $G$ is a {\em finite} graph, the {\em eternal vertex cover problem} consists in finding an eternal vertex cover of minimum cardinality, and this number is denoted with $\evcn(G)$. 
If $G$ is an {\em infinite} grid, the {\em eternal vertex cover problem} consists in finding an eternal vertex cover $S$ such that there exists $n_0 > 0$ and, for every $n \geq n_0$ there are (possibly coinciding) finite rectangular grids of $G$, $G(n,n)$ and $G'(n,n)$ with vertex set $V(n)$ and $V'(n)$, such that $S \cap V(n)$ and $S' \cap V'(n)$ are minimum vertex covers for $G(n,n)$ and $G'(n,n)$, respectively.

In both cases, we will say that a solution of the eternal vertex cover problem is a {\em minimum eternal vertex cover}.
\end{definition}

In order to compare the results we obtain for finite and infinite graphs, we
introduce the definitions of $\rho$ and $\rho^{\infty}$.

\begin{definition}
Let $G$ be a (finite) graph.
We call $\rho$ the ratio between the number of vertices in a minimum vertex cover and the number of all the vertices of $G$, and $\rho^{\infty}$ the ratio between the number of vertices in an eternal minimum vertex cover and the number of all the vertices of $G$. 

Let $G$ be an infinite grid, $C$ and $S$ be a minimum vertex cover and an eternal minimum vertex cover, respectively, for $G$. For each $n >0$ consider every possible finite rectangular grid $G(n,n)$ and let $V(n)$ be its vertex set; compute the minimum over all finite grids of the ratio between $|V(n) \cap C|$ (respectively $|V(n) \cap S|$) and $|V(n)|=n^2$.
We call $\rho$ (respectively $\rho^{\infty}$) the limit as $n$ goes to $\infty$ of this ratio.
\end{definition}

Intuitively, $\rho$ and $\rho^{\infty}$ represent the fraction of vertices that belong to a minimum cardinality cover.

In order to ease the exposition of the proofs, when handling infinite grids and when the cover is constituted by a pattern that is replicated in all directions along the grid, it is very easy to compute $\rho$ ($\rho^{\infty}$) as the portion of vertices included in the cover w.r.t.\ the total number of vertices in the pattern.
So, for example, if the cover includes alternate vertices along both rows and columns, $\rho$ ($\rho^{\infty}$) will be 1/2. More in general, if the cover includes $i$ vertices out of $i+1$ along both rows and columns, $\rho$ ($\rho^{\infty}$) will be $\frac{i}{i+1}$.
We will exploit this observation for quickly computing many upper bounds on $\rho$ ($\rho^{\infty}$) of infinite grids.

\section{Vertex Cover in (Infinite and Finite) Regular Grids}

In view of the result that \(\mvc \le \evcn \le 2 \mvc\)~\cite{klostermeyer2009edge}, it trivially holds also that $\rho \leq \rho^{\infty} \leq 2 \rho$, so we first determine $\rho$ for the graphs we handle, in order to have immediate lower and upper bounds for $\rho^{\infty}$, and then we directly focus on $\rho^{\infty}$.

\subsection{Infinite Regular Grids}

Preliminarily, observe that all the considered (infinite) grids have regular degree $\Delta$
(3 for the hexagonal grids, 4 for the squared grids, 6 for the triangular grids, 8 for the octagonal grids).
An obvious lower bound on $\rho$ is given by the case in which every edge is covered by exactly one vertex and hence is, in all cases, $\rho \geq 1/2$.

For what concerns the hexagonal and squared grid, this trivial lower bound on $\rho$ easily matches with the upper bound obtained by including in the vertex cover {\em e.g.} all the vertices at even $y$-coordinate (see Fig.~\ref{fig:rho-infinite-hex}) and all the vertices at even sum of coordinates (see Fig.~\ref{fig:rho-infinite-square}), respectively, so proving the following result:

\begin{theorem}
\label{th.vc_infiniteT3}
$\rho(T_3)=\rho(T_4)=\frac{1}{2}$.
\end{theorem}

Concerning the triangular grid, a different value can be proved:

\begin{figure}[ht]
\centering
\subfloat[Hexagonal grid\label{fig:rho-infinite-hex}]{
\includegraphics[scale=0.8]{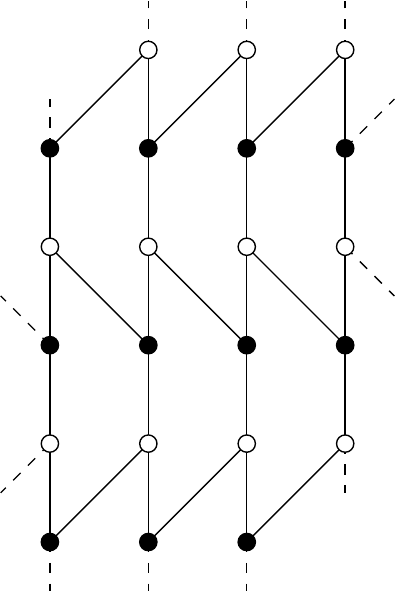}
}~
\subfloat[Squared grid\label{fig:rho-infinite-square}]{
\includegraphics[scale=0.8]{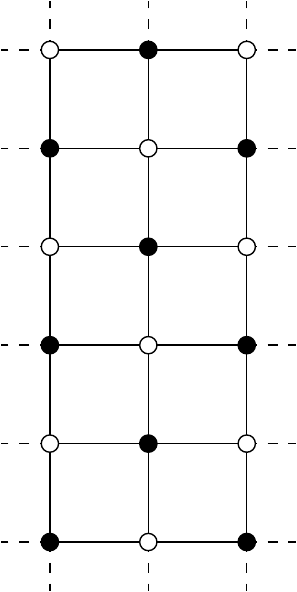}
}~
\subfloat[Triangular grid\label{fig:rho-infinite-tri}]{
\includegraphics[scale=0.8]{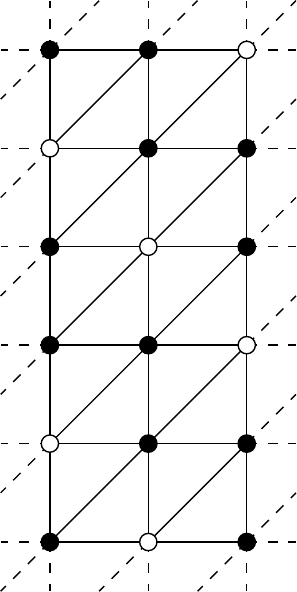}}
~
\subfloat[Octagonal grid\label{fig:rho-infinite-octa}]{
\includegraphics[scale=0.8]{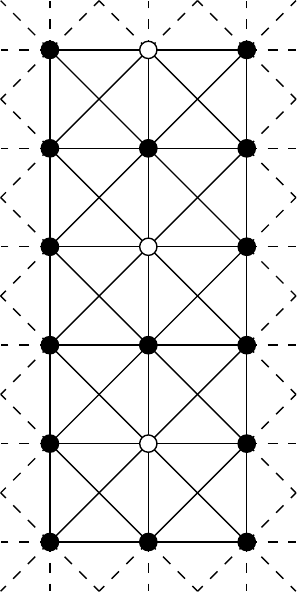}}
\caption{Examples of minimum vertex covers in infinite regular grids (black vertices are within the vertex cover).
}
\end{figure}

\begin{theorem}
\label{th.vc_infiniteT6}
$\rho(T_6)=\frac{2}{3}$.
\end{theorem}

\begin{proof}
First we prove the lower bound by showing that, for any finite rectangular grid $T_6(n,n)$, any vertex cover must contain at least 2/3 of the vertices, {\em i.e.} $C \cap V(n) \geq 2/3$.

Consider a minimum vertex cover $C$ and any row of the triangular grid whose vertices (at $y$-coordinate $\bar{y}$) induce an infinite path $P$; observe that it is not possible that two consecutive vertices of $P$ lie outside $C$ because the complement of every vertex cover must be an independent set.

Let $P'$ be the subgraph induced by the vertices at $y$-coordinate $\bar{y}-1$ ({\em i.e.} the row immediately under $P$). 
In the following we will call two vertices with the same name adding to one of them the apex symbol to mean that they lie one on $P$ and the other one on $P'$ and have the same $x$-coordinate.

Refer to Figure~\ref{fig:pattern}.
We proceed by considering $i \geq 1$ consecutive vertices on $P$, $v_1, \ldots v_i$ that lie in $C$ and are at $x$-coordinate $\bar{x}+1, \ldots , \bar{x}+i$; 
the previous and next vertices, $u_2$ and $w_2$, at $x$-coordinate $\bar{x}$ and $\bar{x}+i+1$, respectively do not;
finally, the further previous and next vertices, $u_1$ and $w_1$, at $x$-coordinate $\bar{x}-1$ and $\bar{x}+i+2$, respectively, are again in $C$.
We will show that, for any $i \geq 1$, the corresponding finite portion of $T_6$ with $2(i+4)$ vertices necessarily has at least 2/3 of its vertices in $C$.

\begin{figure}[ht]
\hspace*{-.5cm}
\subfloat[]{
\includegraphics[scale=0.65]{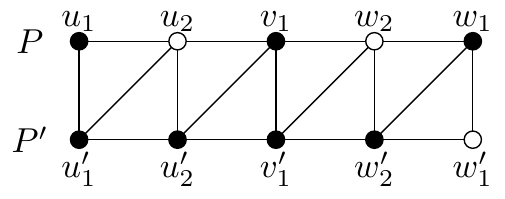}\label{fig:pattern.a}
}~~~~
\subfloat[]{
\includegraphics[scale=0.65]{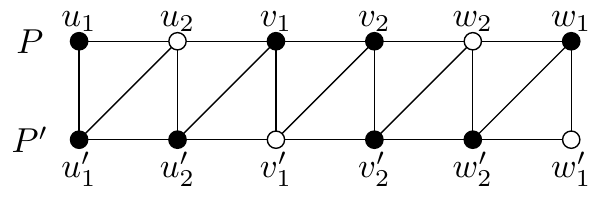}\label{fig:pattern.b}
}
\subfloat[]{
\includegraphics[scale=0.65]{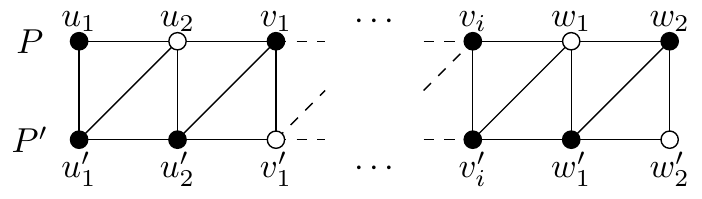}\label{fig:pattern.c}
}
\caption{Possible cases in the proof of Theorem~\ref{th.vc_infiniteT6}.}
\label{fig:pattern}
\end{figure}

If $i=1$ (see Fig.~\ref{fig:pattern.a}),
in view of this configuration, necessarily $u_1'$, $u_2'$, $v'$ and $w_2'$ must lie in $C$ in order to guarantee complete coverage of the edges. 
Looking at the considered 10 vertices, it follows that at least $7/10>2/3$ of the vertices are in $C$.
It is worth noting that if we look at a different subset of vertices ({\em e.g.} vertices from $u_2$ to $w_1$ and from $u'_2$ to $w'_2$) we obtain a fraction lower than 2/3 (in the example, $5/8 < 2/3$); nevertheless, the infinity of the grid guarantees that all the 10 vertices are present, and the correctness of the reasoning follows.

If $i=2$ (see Fig.~\ref{fig:pattern.b}), it holds that $u_1$, $v_1, v_2$ and $w_1$ are in $C$, while $u_2$ and $w_2$ are outside it.
On $P'$, $u_1', u_2', v_2'$ and $w_2'$ must necessarily be in $C$. So at least 8 out of 12 ({\em i.e.} 2/3) of the vertices are in $C$.

For any $i \geq 3$ (see Fig.~\ref{fig:pattern.c}),on $P'$, $u_1', u_2'$, $w_2'$ and $v_i'$ must lie in $C$.
Moreover, observe that $i-2$ horizontal edges remain uncovered, so at least further $\lceil (i-2)/2\rceil$ vertices must enter in $C$.
Looking at the $2(i+4)$ considered vertices, at least $i+6+\lceil (i-2)/2\rceil$ vertices must lie in $C$, leading to a ratio $\geq 3/4 -7/20 > 2/3$.

Observe that this reasoning leads to the result that $\rho(T_6) \geq 2/3$ indeed,
consider any rectangular grid $T_6(n,n)$ with $n \geq 5$ and let $\bar{\imath}$ the minimum integer such that $T_6(n,n)$ contains such a subgraph with $2(\bar{i}+4)$ vertices; if $\bar{\imath}$ exists finite, then clearly $\rho(T_6(n,n)) \geq 2/3$; if, on the contrary, $\bar{\imath}$ does not exist, then it means that each row (except the last one) contains at least $n-2$ vertices of $C$ and anyway $\rho(T_6(n,n)) \geq 2/3$.
Since the case $i=2$ achieves the lowest ratio, the previous proof is constructive and allows us to design a configuration for which this lower bound is achieved: we include in the vertex cover all vertices except the ones at coordinates $(x,y)$ s.t.\ if $x \mod 3=0$ then $y=3k+1$, if $x \mod 3=1$ then $y=3k$ and if $x \mod 3=2$ then $y=3k+2$ for each $k \in \mathbb{Z}$  (see Fig.~\ref{fig:rho-infinite-tri}), so proving the statement.
\end{proof}

\begin{theorem}
\label{th.vc_infiniteT8}
$\rho(T_8)=\frac{3}{4}$.
\end{theorem}

\begin{proof}
For what concerns the lower bound, we proceed exactly in the same way as in the previous theorem, with the only difference that, when on $P$ there is a vertex outside $C$, all its three adjacent vertices on $P'$ must belong to $C$.
This leads to say that $\rho(T_8) \geq 3/4$.

A feasible vertex cover is obtained by putting in it all the vertices at odd $x$-coordinates and alternated vertices at even $x$-coordinates (see Fig.~\ref{fig:rho-infinite-octa}).
Since this configuration includes in $C$ 3/4 of the vertices, it is optimum.
\end{proof}

\subsection{Finite Regular Grids}

Observe that finite grids are not regular graphs anymore, and so the lower bounds on $\mvc$ for the infinite case could be not valid anymore.

More in detail, for what concerns the squared grid, 
we can partition its $m$ edges into two sets, the $m'$ edges that lie on the external face (and are incident to vertices of degree either 2 or 3) and the $m-m'$ edges that are incident to a vertex of degree 4. 
Let be given any vertex cover $C$.
At least $m'/2$ vertices are necessarily in $C$ to cover all the first set of edges, call $C' \subseteq C$ this set; the vertices of $C'$ cover also at least $|C'|-4$ and at most $|C'|$ edges not lying on the external face (indeed corner vertices are incident only to boundary edges). So, at least
$(m-m'-|C'|)/\Delta \geq (m-m'-m'/2)/4$ further vertices must be added to cover the second set. 
We hence have that $\mvc(T_4(h,w)) \geq m'/2+1/4(m-3m'/2)=m/4+m'/8$.
Substituting the values of $m=2wh-h-w$ and $m'=2(h-1)+2(w-1)$, we get that $\mvc(T_4(h,w)) \geq (wh-1)/2$ implying that $\rho(T_4(h,w)) \geq 1/2-1/(2hw)$.
Note that this result is assumed as known in~\cite{klostermeyer2009edge} but we could not find any proof in the literature, so we provided it here for the sake of completeness.

\medskip

We now study the hexagonal grid.
First observe that $h \geq 4$ must always be even and $w \geq 2$. 
Consider the subgraph of $T_3(h,w)$ induced by two consecutive columns of vertices, $T_3(h,2)$, and a cover $C$ on it; it has $2(h-1)$ vertices organized on $h/2-1$ hexagons; all the $2(h-1)$ edges lying on the external face must be covered including in $C$ at least half of the vertices; this remains true for every subgraph of this kind;
this proves that also $\mvc(T_3(h,w)) \geq hw/2$.

\medskip

For what concerns the triangular and octagonal grids, we cannot use anymore the idea of considering two consecutive rows (as we did in the infinite case), because the finiteness of the graph could not guarantee the correctness of the reasoning.

\medskip

The upper bounds on $\rho$ for all the four regular grids can be inherited by the corresponding infinite grids by simply observing that, in view of the definition of finite  grids, any minimum vertex cover $C$ of an infinite grid induces also a vertex cover (not necessarily minimum) on a finite (rectangular) subgraph.

\medskip

In the following theorem we summarize the previous results and complete them with the uncovered cases:

\begin{theorem}
\label{th.finite_vc_T3_T6}
$\frac{1}{2}-\frac{1}{2hw} \leq \rho(T_4)(h,w) \leq \frac{1}{2}$ for any $h, w \geq 2$ and $\rho(T_3)(h,w)=\frac{1}{2}$ for any $h \geq 4$ and $w \geq 2$.

\noindent
$\frac{2}{3}-\frac{1}{3w} \leq \rho(T_6(h,w))  \leq \frac{2}{3}$ for any $h \geq 2$ and $w \geq 2$.

\noindent
$\frac{3}{4}-\frac{1}{4w} \leq \rho(T_8(h,w)) \leq \frac{3}{4}$ for $h \geq 2$ and $w \geq 2$.
\end{theorem}
\begin{proof}
The exact value of $\rho(T_4(h,w))$ and $\rho(T_3(h,w))$ comes from the observations above.
It remains to prove only lower and upper bounds on $\rho(T_6(h,w))$ and $\rho(T_8(h,w))$. 

Let us deal with $\rho(T_6(h,w))$ first. 
Consider the subgraph of $T_6(h,w)$, $T_6(2,w)$, induced by any two consecutive rows of vertices and reduce to study case $h=2$; the number of vertices of this subgraph is $2w$.

If $w \geq 2$, 
Let $v$ and $v'$ be the rightmost vertices of $T_6(2,w)$, $u$ and $u'$ be the vertices at their immediate left, and let  $v$ and $u$ belong to the upper row.

Let $C$ any given vertex cover for $T_6(2,w)$.
We will prove that  $\rho(T_6(2,w)) \geq \frac{2}{3}$ if both $v$ and $v'$ belong to $C$, $\rho(T_6(2,w)) \geq \frac{2}{3}-\frac{1}{6w}$ if only one between $v$ and $v'$ is in $C$ but both $u$ and $u'$ are in $C$, while $\rho(T_6(2,w)) \geq \frac{2}{3}-\frac{1}{3w}$ if one between $v$ and $v'$  and one between $u$ and $u'$ belong to $C$. (Note that it cannot be that both $v$ and $v'$ are outside $C$, otherwise edge $\{ v,v' \}$ would remain uncovered.)

We proceed by induction on $w$.
The basis of the induction is represented by special cases $T_6(2,2)$ and $T_6(2,3)$ for which the claim is easily true since $\rho(T_6(2,2))=\frac{1}{2}=\frac{2}{3}-\frac{1}{3 \cdot 2}$ and $\rho(T_6(2,3))=\frac{2}{3}$ (See Fig.~\ref{fig:B.a} and~\ref{fig:B.b}).

\begin{figure}[ht]
\centering
\subfloat[\label{fig:B.a}]{
\includegraphics[scale=0.7]{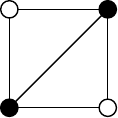}
}~
\subfloat[\label{fig:B.b}]{
\includegraphics[scale=0.7]{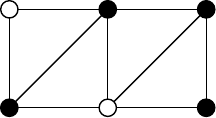}
}~
\subfloat[\label{fig:B.c}]{
\includegraphics[scale=0.7]{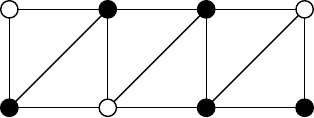}
}~
\subfloat[\label{fig:B.e}]{
\includegraphics[scale=0.7]{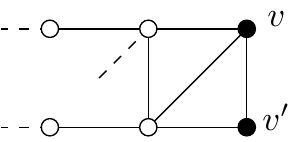}
}~
\subfloat[\label{fig:B.f}]{
\includegraphics[scale=0.7]{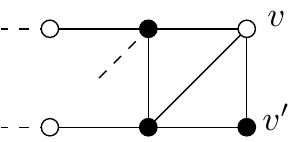}
}~
\subfloat[\label{fig:B.g}]{
\includegraphics[scale=0.7]{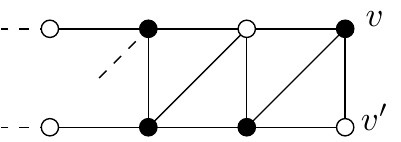}
}
\caption{Possible cases in the proof of Theorem~\ref{th.finite_vc_T3_T6}.}

\end{figure}

Given $T_6(2,w)$, $w \geq 4$, it trivially contains $T_6(2,w-1)$, obtained by simply eliminating $v$ and $v'$ from it.
Consider a vertex cover $C$ for $T_6(2,w)$; it induces a vertex cover $C'$ on $T_6(2,w-1)$ and for it the inductive hypothesis holds.

Now, consider $v$ and $v'$.
If they both belong to $C$ (see Fig.~\ref{fig:B.e}), we have no information about the rightmost vertices of $T_6(2,w-1)$, so we use the worst inductive hypothesis $\frac{|C'|}{2(w-1)} \geq \frac{2}{3}-\frac{1}{3(w-1)}$ from which we deduce $|C'|\geq \frac{4}{3}w-2$.
So:
$$
\rho(T_6(2,w))=\frac{|C|}{2w}=
\frac{|C'|+2}{2w} \geq 
\frac{\frac{4}{3}w-2+2}{2w}=\frac{2}{3}.
$$

If, on the contrary, only $v'$ belongs to $C$, it is easy to see that  $u$ and $u'$ must necessarily be in $C$ (see Fig.~\ref{fig:B.f}).
In this case, $|C'| \geq \frac{4}{3}w-\frac{4}{3}$
and we get:
$$
\rho(T_6(2,w))=\frac{|C|}{2w}=
\frac{|C'|+1}{2w} \geq 
\frac{\frac{4}{3}w-\frac{4}{3}+1}{2w}\geq \frac{2}{3}-\frac{1}{6w} \geq \frac{3}{5}+\frac{2}{5w}.
$$

The same holds if $v$, $u$ and $u'$ belong to $C$ while $v'$ does not.

If, finally, $v$ belongs to $C$ while $v'$ does not, and only one between $u$ and $u'$ belong to $C$, it must necessarily be $u'$ (otherwise $C$ would not be a vertex cover), so that the two vertices to the immediate left must belong to $C$ (see Fig.~\ref{fig:B.g}).
In such a case we exploit the inductive hypothesis on $T_6(2,w-2)$ and its induced vertex cover $C''$ whose cardinality is $|C''| \geq \frac{4}{3}w-\frac{8}{3}$:
$$
\rho(T_6(2,w))=\frac{|C|}{2w}=\frac{|C''|+2}{2w} \geq
\frac{\frac{4}{3}w-\frac{8}{3}+2}{2w} \geq \frac{2}{3}- \frac{1}{3w}.
$$

Observe that all these three lower bounds are tight since they are achieved, for example, for $w=2,3,4$ (see Fig.~\ref{fig:B.c}).

Since every $T_6(h,w)$ with $h \geq 2$ and $w \geq 5$ contains a $T_6(2,w)$, this ratio continues to hold for every finite triangular grid.

For what concerns the upper bound, an arrangement alternating a pair of vertices in $C$ with a single vertex outside $C$ on each column (see Fig.~\ref{fig:rho-infinite-tri})
covers the whole graph and the corresponding value of $\rho$ is upper bounded by 2/3, if we opportunely shift this configuration in order to maximize the number of vertices not belonging to $C$. 

\medskip

Finally, we deal with $\rho(T_8(h,w))$.

We replicate here the idea of reducing to $T_8(2,w)$ and performing an inductive proof by eliminating from the graph the rightmost vertices $v$ and $v'$.
Our claim is that $\rho(T_8(2,w)) \geq 3/4$ if both $v$ and $v'$ belong to $C$ while $\rho(T_8(2,w)) \geq 3/4- 1/(4w)$ otherwise.

The basis of induction is the special case $T_8(2,2)$ for which $\rho(T_8(2,2))=3/4$. 
For any $w \geq 3$, we consider vertex cover $C'$ deduced by $C$ on the subgraph of $T_8(2,w)$ induced by the $2(w-1)$ leftmost vertices.

If both $v$ and $v'$ belong to $C$, we do not know anything about the vertices at their immediate left, so by inductive hypothesis, $\frac{|C'|}{2(w-1)} \geq 3/4 -\frac{1}{4(w-1)}$ from which it outcomes that $|C'| \geq 3/2w-2$.
Then:
$$
\rho(T_8(2,w))=\frac{|C|}{2w}=\frac{|C'|+2}{2w}\geq \frac{3}{4}.
$$

If, on the contrary, only one between $v$ and $v'$ belong to $C$, then the vertices at their immediate left must both be in $C$, hence $|C'| \geq \frac{3}{2}(w-1)$.
Hence:
$$
\rho(T_8(2,w))=\frac{|C|}{2w}=\frac{|C'|+1}{2w}\geq \frac{\frac{3}{2}(w-1)+1}{2w}=\frac{3}{4}+\frac{-\frac{3}{2}+1}{2w} \geq \frac{3}{4}-\frac{1}{4w}.
$$
The observation that every $T_8(h,w)$ with $h,w \geq 2$ contains a $T_8(2,w)$ and that the ratio is kept concludes the lower bound proof.

An upper bound of 3/4 is achieved by alternating a column where every vertex is in $C$ with a column where vertices are alternately inside and outside $C$ (see Fig.~\ref{fig:rho-infinite-octa}). 
Indeed, if $w$ is even, there is the same number of columns with all vertices in $C$ and with vertices alternately in $C$; if, vice-versa, $w$ is odd, the number of columns with all vertices in $C$ is one less than the number of columns with vertices alternately in $C$.
\end{proof}

\section{Eternal Vertex Cover in (Infinite and Finite) Regular Grids}

Once we know the value of the vertex cover number of the considered infinite and finite grids, in this section we study their eternal vertex cover number. 

\subsection{Infinite Regular Grids}

For what concerns hexagonal and squared grids, since an optimum vertex cover $C$ for these grids includes alternated vertices, and since the complement of $C$ is a minimum vertex cover as well, such a set constitutes also an eternal minimum vertex cover. 
Indeed, consider squared grids first: after each attack, the guard lying on the vertex whose incident edge has been attacked will traverse that edge, and all the other guards will shift in the same direction. 

In the case of hexagonal grids, we do the same if the attacked edge is vertical; if, on the contrary, the attacked edge is slanted, we shift all the guards lying on an extreme of a slanted edge to the other extreme.
These strategies allow us to pass from a minimum vertex cover $C$ to its complement, which is another minimum vertex cover. 
The infinity of the grid guarantees that any defense step can always be successfully performed.

Also in the case of triangular grids, in which we start from a configuration alternating on each row pairs of vertices inside the vertex cover with single vertices outside it,
whenever an edge is attacked, all guards shift in the same direction as the guard defending the attacked edge.
It is easy to see that the infinity of the grid guarantees that any defense step can always be successfully performed.

Finally, also in the case of octagonal grids, it is easy to transform an optimum vertex cover into another one by simply shifting all the guards in the same direction as the guard defending the attacked edge.

Hence, the following statement can be deduced:

\begin{theorem}
\label{th.infinite_EVC_T3_T4_T6}
$\rho^{\infty}(T_3)=\rho^{\infty}(T_4)=\frac{1}{2}$; 
$\rho^{\infty}(T_6)=\frac{2}{3}$ and $\rho^{\infty}(T_8)=\frac{3}{4}$.
\end{theorem}

\subsection{Finite Squared and Hexagonal grids}

When passing from infinite to finite grids, an upper bound on $\rho^{\infty}(T_4(h,w))$ becomes $\frac{1}{2}+\frac{1}{2hw}$ thanks to a technique consisting in computing a Hamiltonian cycle passing through the attacked edge (which can be found in polynomial time on squared grids) and moving all the guards along the cycle on in the same direction~\cite{klostermeyer2009edge}.

Unfortunately, the problem of finding a Hamiltonian cycle in hexagonal grids is proven to be NP-complete~\cite{arkin2009not} and so we cannot exploit the same technique.
We propose another defense strategy that proves the following result.

\begin{theorem}
\label{th.finite_EVC_T3}
$\rho^{\infty}(T_3(h,w))=\frac{1}{2}$ for any $h \geq 4$ and $w \geq 2$.
\end{theorem}

\begin{proof}
Let \(S\) the subset of vertices of $T_3(h,w)$ in which the guards are placed.
$S$ includes alternated vertices, for any $x$-coordinate, for example at odd $y$-coordinates. So its cardinality is exactly equal to the half of the vertices of the graph, due to the parity of $h$.

Given any attacked edge, our strategy will shift all the guards in the graph.
Namely, all the guards, except for the ones lying on the two columns adjacent to the attacked edge, will shift one position down (or respectively up)  w.r.t.\ their original one.
As an invariant of the strategy, after every attack and defense, guards lie anyway on alternated vertices for any $x$-coordinate (either at odd or at even $y$-coordinates). 
This guarantees the correctness of our strategy.

More formally, let the origin of the Cartesian plane coincide with the lower left corner (that is at coordinates $(0,0)$) and, w.l.o.g., let the embedding be such that there is always a vertex at point $(0,0)$.

It is not restrictive to assume that guards lie on every vertex with odd $y$-coordinates (otherwise we can slightly modify the following strategy in order to let it work when guards lie at even $y$-coordinates).

Let the attacked edge be \(\{(x',y'),(x'',y'')\}\), where $(x',y')$ always contains the guard while $(x'',y'')$ does not; 
$x''$ is either $x'$ or $x' \pm 1$ and  $y''$ is $y'\pm 1$.

The strategy is the following:
\begin{itemize}
\item every guard in \((x,y)\) for $x < x'$ and every guard in \((x,y)\) for $x>x'+1$ and $y \neq 1, h-1$ move down, i.e., \(f((x,y)) = (x,y-1)\);

\item every guard in \((x,h-1)\) for  $x>x'+1$ move down on the adjacent column, {\em i.e.} either \(f((x,h-1)) = (x+1,h-2)\) or \(f((x,h-1)) = (x-1,h-2)\) (according to the value of $h$);

\item every guard in \((x,1)\) for  $x>x'+1$ move down on the left adjacent column, {\em i.e.} \(f((x,1)) = (x-1,0)\);
\end{itemize}

For the remaining guards, {\em i.e.} the ones placed at $x$-coordinates $x'$ and $x''$, we have
different cases:

\begin{itemize}
\item if the attacked edge is \(\{(x',y'),(x',y'+1)\}\) (see Figure~\ref{fig:hex-defense-strategy}a)
every guard in \((x',y)\) for \(y \in [1, \ldots, h-2]\) moves up, {\em i.e.} \(f((x',y)) = (x',y+1)\) and every guard in \((x'+1,y)\) for \(y \in [2, \ldots, h-2]\) moves down, {\em i.e.} \(f((x'+1,y)) = (x'+1,y-1)\), finally, guard in position \((x'+1,1)\) moves down to the left adjacent column, {\em i.e.} \(f((x'+1,1)) = (x',0)\) and guard in position \((x',h-1)\) moves down to the right adjacent column, {\em i.e.} \(f((x',h-1)) = (x'+1,h-2)\);

\item if the attacked edge is \(\{(x',y'),(x',y'-1)\}\) (see Figure~\ref{fig:hex-defense-strategy}b)
every guard in \((x',y)\) for \(y \in [1, \ldots, h-1]\) moves down, {\em i.e.} \(f((x',y)) = (x',y-1)\) and every guard in \((x'+1,y)\) for \(y \in [1, \ldots, h-2]\) moves up, {\em i.e.} \(f((x'+1,y)) = (x'+1,y+1)\), finally, guard in position \((x'+1,h-1)\) moves down to the right adjacent column, {\em i.e.} \(f((x'+1,h-1)) = (x'+2,h-2)\);

\item if the attacked edge is \(\{(x',y'),(x'+1,y'-1)\}\) (see Figure~\ref{fig:hex-defense-strategy}c)
every guard in \((x',y)\) for \(y \in [1, \ldots, y'-1]\) moves up, {\em i.e.} \(f((x',y)) = (x',y+1)\) and every guard in \((x',y)\) for \(y \in [y'+1, \ldots, h-1]\) moves down, {\em i.e.} \(f((x',y)) = (x',y-1)\).
Every guard in \((x'+1,y)\) for \(y \in [2, \ldots, y'-2]\) moves down, {\em i.e.} \(f((x',y)) = (x',y-1)\) and every guard in \((x',y)\) for \(y \in [y', \ldots, h-2]\) moves up, {\em i.e.} \(f((x',y)) = (x',y+1)\).
Finally, guard in position \((x'+1,h-1)\) moves down to the right adjacent column, {\em i.e.} \(f((x'+1,h-1)) = (x'+2,h-2)\), guard in position \((x'+1,1)\) moves down to the left adjacent column, {\em i.e.} \(f((x'+1,1)) = (x',0)\) and guard in position \((x',y')\) moves along the attacked edge.

\item if, finally, the attacked edge is \(\{(x',y'),(x'-1,y'-1)\}\)
(see Figure~\ref{fig:hex-defense-strategy d})
every guard in \((x'-1,y)\) for \(y \in [1, \ldots, y'-2]\) moves down, {\em i.e.} \(f((x'-1,y)) = (x'-1,y-1)\) and every guard in \((x'-1,y)\) for \(y \in [y', \ldots, h-2]\) moves up, {\em i.e.} \(f((x',y)) = (x',y+1)\).
Every guard in \((x',y)\) for \(y \in [1, \ldots, y'-1]\) moves up, {\em i.e.} \(f((x'-1,y)) = (x'-1,y+1)\) and every guard in \((x'-1,y)\) for \(y \in [y'+1, \ldots, h-2]\) moves down, {\em i.e.} \(f((x',y)) = (x',y-1)\).
Finally, guards in position  \((x'-1,h-1)\) and \((x',h-1)\) move down to the right adjacent column, {\em i.e.} \(f((x'-1,h-1)) = (x',h-2)\) and \(f((x',h-1)) = (x'+1,h-2)\), guard in position \((x',y')\) moves along the attacked edge.
\end{itemize}
\end{proof}

\begin{figure}[ht]
\hspace*{-.8cm}
\subfloat[]{
\includegraphics[scale=0.65]{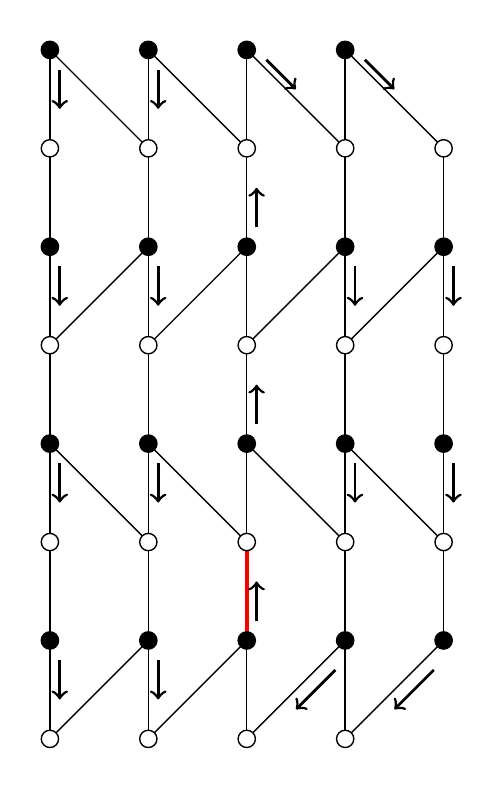}
}~
\subfloat[]{
\includegraphics[scale=0.65]{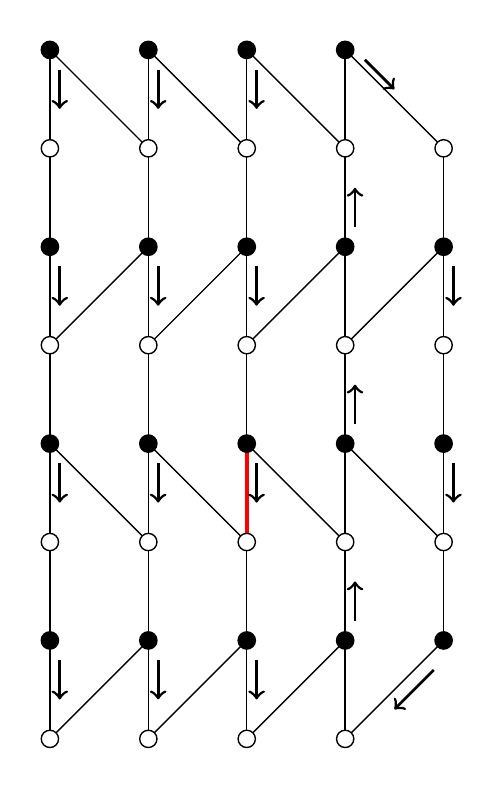}
}~
\subfloat[]{
\includegraphics[scale=0.65]{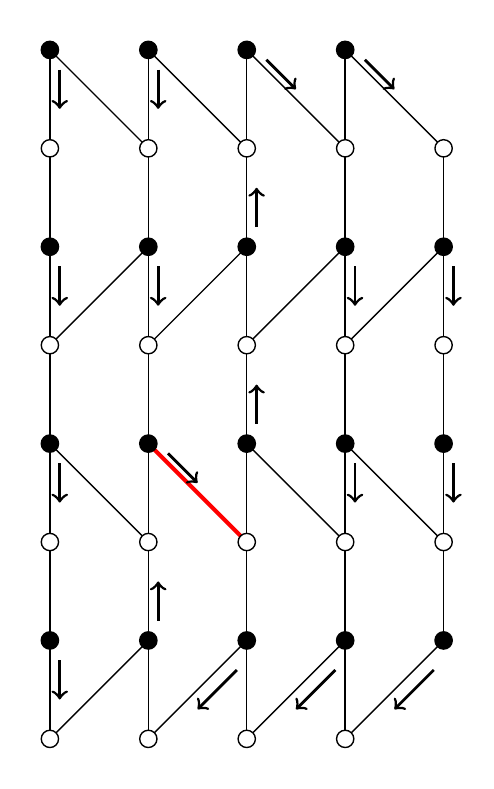}
}~
\subfloat[\label{fig:hex-defense-strategy d}]{
\includegraphics[scale=0.65]{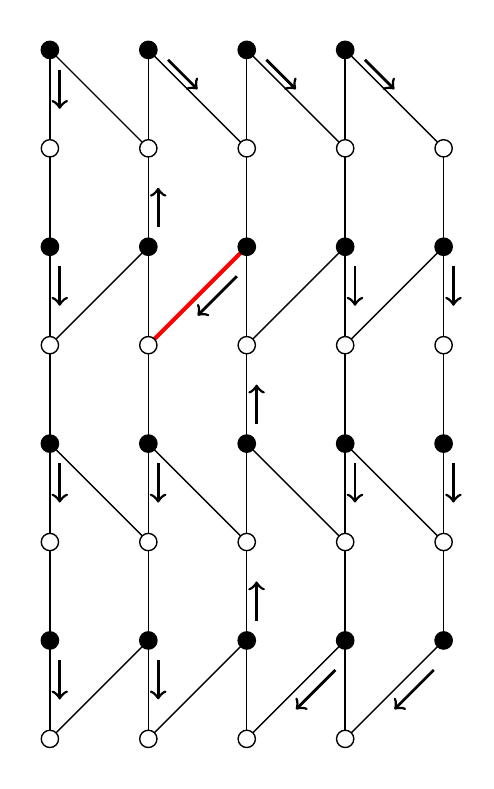}
}
\caption{
Possible cases in the proof of Theorem~\ref{th.finite_EVC_T3}.
}
\label{fig:hex-defense-strategy}
\end{figure}

\subsection{Finite Triangular grids}

Also in this case we cannot exploit the idea used in~\cite{klostermeyer2009edge} for the squared grid based on the Hamiltonian cycle, because on finite triangular grids this approach could bring two vertices without any guard to become adjacent.

Hence, we propose the defense strategy in the following proof that requires guards on approximately 2/3 of the vertices. 

\begin{theorem}
\label{th.finite_EVC_T6}
$\rho^{\infty}(T_6(h,w)) \le \frac{2}{3}+\frac{4}{h}+\frac{4}{h^2}$ for any $w \geq h \geq 3$, $\rho^{\infty}(T_6(2,w)) \le \frac{2}{3} +\frac{2}{w}$ for any $w \geq 3$ and $\rho^{\infty}(T_6(2,2)) = \frac{3}{4}$.
\end{theorem}

\begin{proof}
It is clear that, when $h=w=2$, two vertices in the eternal vertex cover are not enough, so we have to take at least three, and they are always able to defend $T_6(2,2)$.

\medskip

When $w \geq h=2$, let $w=3w'+w''$, $w'' \in \{0, 1, 2 \}$ and let the vertices in $T_6(2,w)$ be at $y$-coordinates either $0$ or $1$ and at $x$-coordinates from $0$ to $w-1$.
Consider the following two vertex sets: \\
$S_1=\{ v \mbox{ at coordinates } (1+3i, 0) \mbox{ and } (2+3i, 0) \mbox{ and } (3i, 1) \mbox{ and } (1+3i, 1) , i=0, \ldots , w'-1\} \cup \{ v \mbox{ at coordinates } (3w'+j, 0) \mbox{ and } (3w'+j,1), j=0, \ldots , w''-1 \}$ (see Fig.~\ref{fig:tri grid 2w config 1})\\
and \\
$S_2=\{ v \mbox{ at coordinates } (3i,0) \mbox{ and } (2+3i, 0) \mbox{ and } (1+3i, 1) \mbox{ and } (2+3i, 1) , i=0, \ldots , w'-1\} \cup \{ v \mbox{ at coordinates } (3w'+j, 0) \mbox{ and } (3w'+j, 1), j=0, \ldots , w''-1 \}$ (see Fig.~\ref{fig:tri grid 2w config 2}).

\begin{figure}[ht]
    \centering
    
    \begin{minipage}{0.45\textwidth}
    \begin{minipage}{0.45\textwidth}
    \begin{minipage}{0.45\textwidth}
     \subfloat[\label{fig:tri grid 2w config 1}]{
    \includegraphics[scale=1]{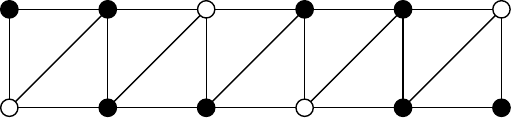}}
    \end{minipage}
    \\
    \begin{minipage}{0.45\textwidth}
    \subfloat[\label{fig:tri grid 2w config 2}]{
    \includegraphics[scale=1]{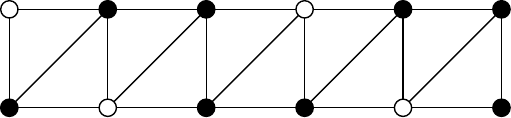}}
    \end{minipage}
    \end{minipage}\\
    \begin{minipage}{0.45\textwidth}
        \subfloat[\label{fig:tri grid tile Q}]{
        \includegraphics[scale=1]{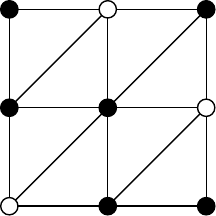}
        }
        ~~~~
        \subfloat[\label{fig:tri grid tile Z}]{
        \includegraphics[scale=1]{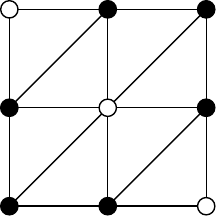}
    }
    \end{minipage}
    \end{minipage}
    ~
    \begin{minipage}{0.45\textwidth}
        \subfloat[\label{fig:big grid tile}]{
        \includegraphics[scale=.5]{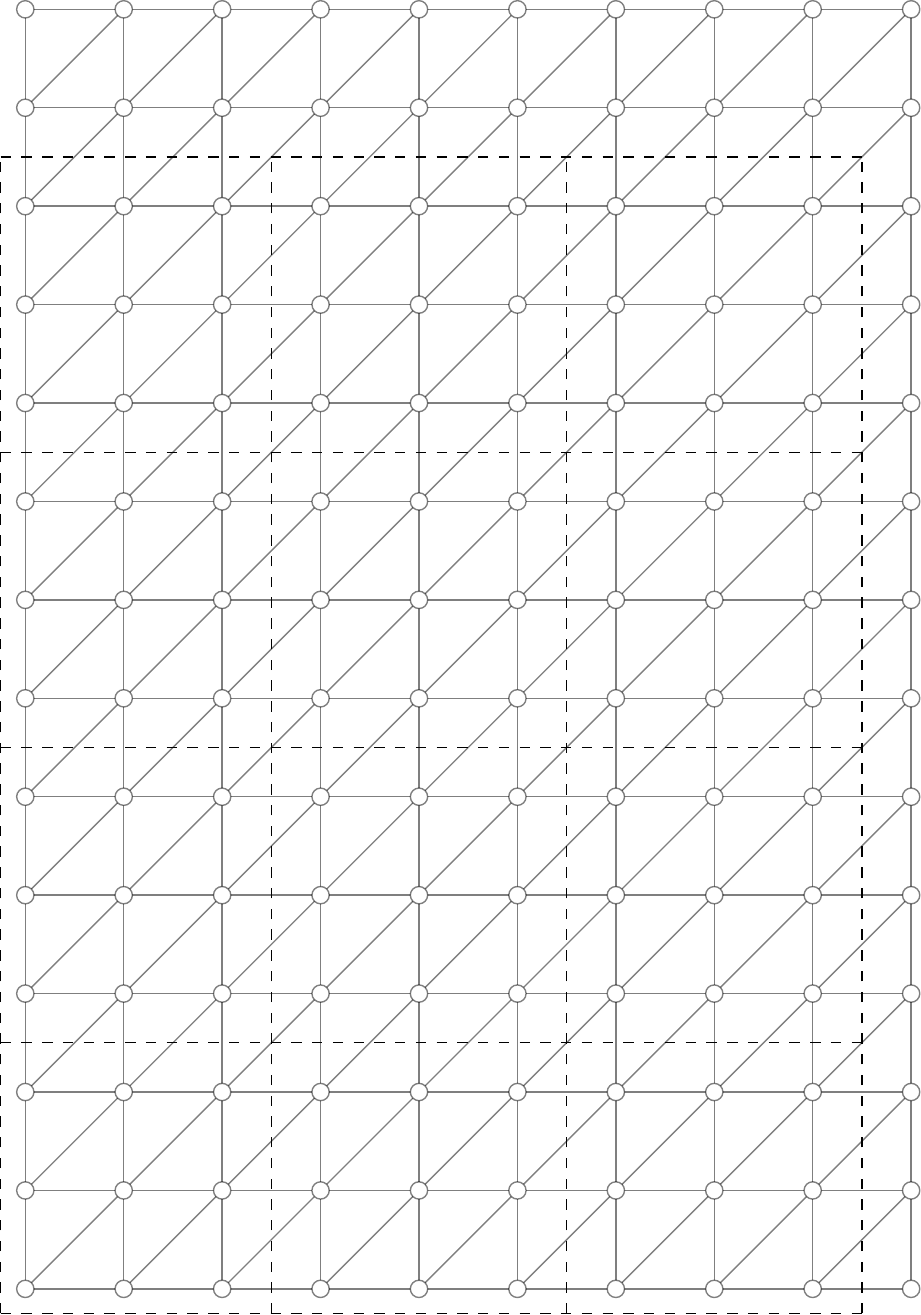}
    }
    \end{minipage}
    \caption{a. $S_1$ when $w\ge h=2$; b. $S_2$ when $w\ge h=2$; c. Tile with $S_1$ when $w,h\ge 3$; d. Tile with $S_2$ when $w,h\ge 3$; e. Partition of $T_6(h,w)$ into tiles plus a stripe.}
    \label{fig:tri grid tile}
\end{figure}

It is trivial to see that both $S_1$ and $S_2$ are vertex covers; moreover, Fig.~\ref{fig:tri 2,w evc} shows different ways to shift guards on the vertices of $S_1$ in order to move them on the vertices of $S_2$.
Note that, for every edge $e$ in $T_6(2,w)$ protected by a single guard, there exists at least one possibility, among those depicted in Fig.~\ref{fig:tri 2,w evc}, shifting a guard along $e$.
Moreover, by simply reversing the direction of the shifts in Fig.~\ref{fig:tri 2,w evc} we get different ways to shift guards on the vertices of $S_2$ in order to move them on the vertices of $S_1$.
In order not to overburden the exposition, we omit the mathematical details, that can be easily deduced by Fig.~\ref{fig:tri 2,w evc}.
It follows that both $S_1$ and $S_2$ are eternal vertex covers for $T_6(2,w)$ and Fig.~\ref{fig:tri 2,w evc} provides a defense strategy transforming one in the other one guaranteeing the possibility to defend an infinite sequence of attacks.

\begin{figure}[ht]
    \centering
    
    \begin{minipage}{0.45\textwidth}
    \includegraphics[width=1\textwidth]{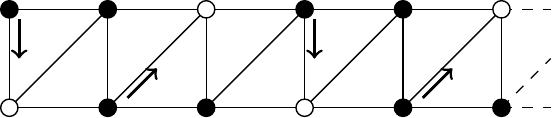}\vspace{5pt}%
    \\
    \includegraphics[width=1\textwidth]{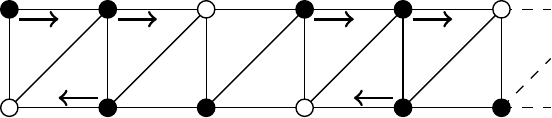}\vspace{5pt}%
    \\
    \includegraphics[width=1\textwidth]{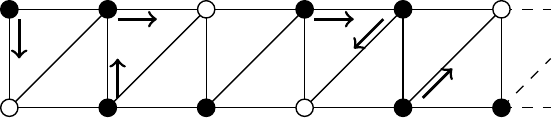}
    \end{minipage}
    ~~
    \begin{minipage}{0.45\textwidth}
    \includegraphics[width=1\textwidth]{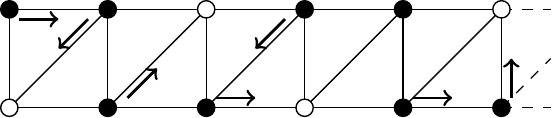}\vspace{5pt}%
    \\
    \includegraphics[width=1\textwidth]{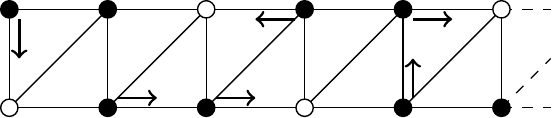}\vspace{5pt}%
    \\
    \includegraphics[width=1\textwidth]{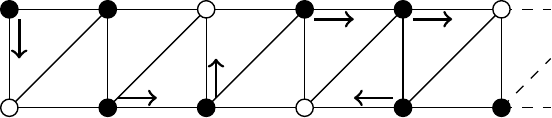}
    \end{minipage}
    
    \caption{Defense strategies for $T_6(2,w)$.}
    \label{fig:tri 2,w evc}
\end{figure}

We conclude this part of proof by computing $\rho^{\infty}(T_6(2,w))$.
$|S_1|=|S_2|=4w'+2w''$ while the number of vertices of $T_6(2,w)$ is $2w$.
Hence, $\rho^{\infty}(T_6(2,w))=
\frac{4w'}{6w'+2w''}+\frac{2w''}{2w}$.
Since $0 \leq w'' \leq 2$ this quantity is upper bounded by $\frac{2}{3}+\frac{2}{w}$.

\medskip

In the general case $w \geq h \geq 3$,
let $h=3h'+h''$ and $w=3w'+w''$. 
So, $T_6(h,w)$ can be partitioned into $h'w'$ $3 \times 3$ tiles plus a stripe at most 2 wide that we can imagine to lie on the upper and rightmost side of the rectangle (see Figure~\ref{fig:big grid tile}).
We now construct two vertex covers $S_1$ and $S_2$ and prove it is possible to shift one into the other one defending every possibly attacked edge, so showing they are eternal for $T_6(h,w)$. 

For each $3 \times 3$ tile, put inside $S_1$ (respectively $S_2$) 6 out of the 9 vertices, all in the same position, as shown in Figure~\ref{fig:tri grid tile Q} (respectively~\ref{fig:tri grid tile Z}).

All the vertices of the stripe are put both in $S_1$ and in $S_2$.

We now consider every possible attacked edge and, for each, we describe a global defense strategy shifting either $S_1$ in $S_2$ or vice-versa. 
More in detail, whichever is the attacked edge, we shift every guard inside every $3 \times 3$ tile from configuration of Figure~\ref{fig:tri grid tile Q} to the one of Figure~\ref{fig:tri grid tile Z} or vice-versa.

Assume to start from guards placed on vertices of $S_1$, first.

If the attacked edge $e$ is inside a tile, look for the picture in Figure~\ref{fig:tri grid 3x3 atk inside tile QZ} with an arrow in correspondence of edge $e$ and use it as global strategy to shift the guards lying on vertices of the cover inside each tile from $S_1$ to $S_2$.

\begin{figure}[ht]
    \centering
    \centering
    \subfloat[\label{fig:tri grid 3x3 atk inside tile QZ}]{
    \begin{minipage}{0.16\textwidth}
    \includegraphics[width=0.9\textwidth]{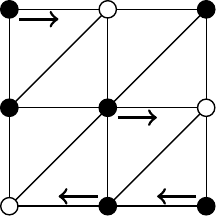}
    \end{minipage}
    \begin{minipage}{0.16\textwidth}
    \includegraphics[width=0.9\textwidth]{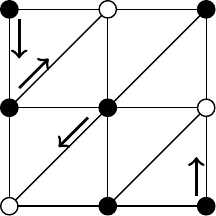}
    \end{minipage}
    \begin{minipage}{0.16\textwidth}
    \includegraphics[width=0.9\textwidth]{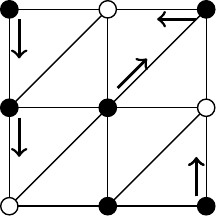}
    \end{minipage}
    \begin{minipage}{0.16\textwidth}
    \includegraphics[width=0.9\textwidth]{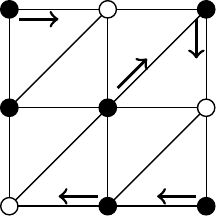}
    \end{minipage}
    \begin{minipage}{0.16\textwidth}
    \includegraphics[width=0.9\textwidth]{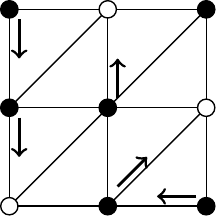}
    \end{minipage}
    }
    \\\hfill
    \subfloat[\label{fig:tri grid 3x3 atk inside tile ZQ}]{
    \begin{minipage}{0.16\textwidth}
    \includegraphics[width=0.9\textwidth]{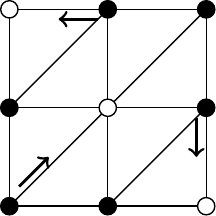}
    \end{minipage}
    \begin{minipage}{0.16\textwidth}
    \includegraphics[width=0.9\textwidth]{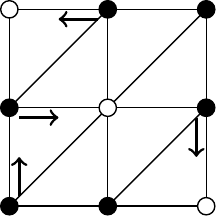}
    \end{minipage}
    \begin{minipage}{0.16\textwidth}
    \includegraphics[width=0.9\textwidth]{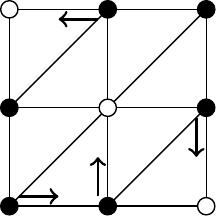}
    \end{minipage}
    \begin{minipage}{0.16\textwidth}
    \includegraphics[width=0.9\textwidth]{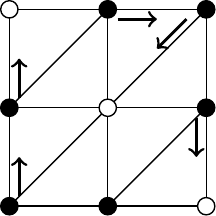}
    \end{minipage}
    \begin{minipage}{0.16\textwidth}
    \includegraphics[width=0.9\textwidth]{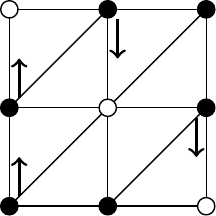}
    \end{minipage}
    \begin{minipage}{0.16\textwidth}
    \includegraphics[width=0.9\textwidth]{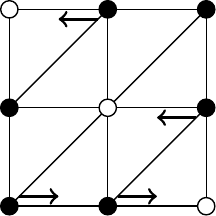}
    \end{minipage}
}
    \caption{a. Shift of guards inside a tile from configuration of Figure~\ref{fig:tri grid tile Q} to the one of Figure~\ref{fig:tri grid tile Z}; b. Shift of guards inside a tile from configuration of Figure~\ref{fig:tri grid tile Z} to the one of Figure~\ref{fig:tri grid tile Q}.}
    \label{fig:tri grid tile pairs attacks}
\end{figure}

If, instead, the attacked edge is vertical and connects two adjacent tiles $t_1$ and $t_2$, then Figure~\ref{fig:tri grid tile pairs QZ} shows how to exchange two guards between $t_1$ and $t_2$ along edges connecting them, and how to move the remaining guards inside $t_1$ and $t_2$; the defense strategy consists in adopting this move while all the other tiles follow anyone among the pictures in Figure~\ref{fig:tri grid 3x3 atk inside tile QZ} to globally shift again from $S_1$ to $S_2$.
If the attacked edge is horizontal, we can use the same figures rotated by 90 degrees, since both configurations in Figure~\ref{fig:tri grid tile Q} and~\ref{fig:tri grid tile Z} are symmetrical w.r.t.\ their diagonal.

\begin{figure}[ht]
    \centering
    \subfloat[\label{fig:tri grid tile pairs QZ}]{
    \begin{minipage}{0.16\textwidth}
\centering
\includegraphics[width=1\textwidth]{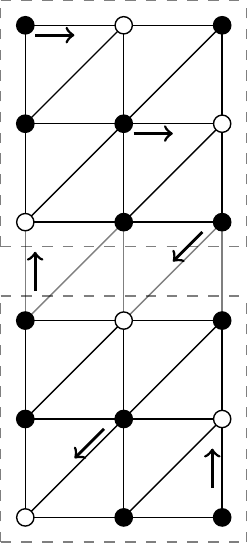}
\end{minipage}~~
\begin{minipage}{0.16\textwidth}
\centering
\includegraphics[width=1\textwidth]{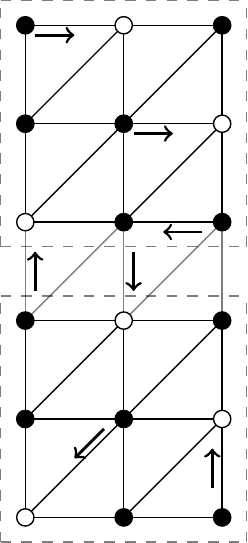}
\end{minipage}~~
    }
    ~~~~
    \subfloat[\label{fig:tri grid tile pairs ZQ}]{
   \begin{minipage}{0.16\textwidth}
\centering
\includegraphics[width=1\textwidth]{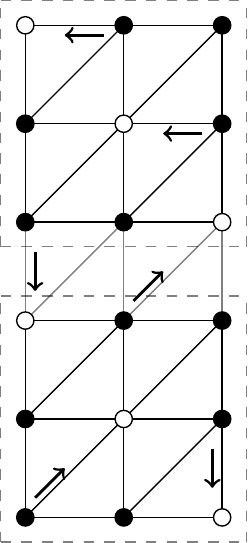}
\end{minipage}~~
\begin{minipage}{0.16\textwidth}
\centering
\includegraphics[width=1\textwidth]{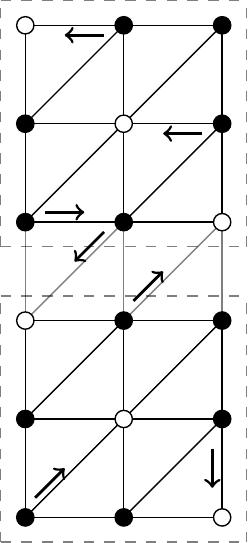}
\end{minipage}~~
}
    \caption{Shift when the attacked edge connects two adjacent tiles.
    a. Shift of guards from configuration of Figure~\ref{fig:tri grid tile Q} to the one of Figure~\ref{fig:tri grid tile Z}; b. Shift of guards from configuration of Figure~\ref{fig:tri grid tile Z} to the one of Figure~\ref{fig:tri grid tile Q}.}
    \label{fig:tri grid tile pairs}
\end{figure}

If, finally, the attacked edge connects a tile with the stripe, Figure~\ref{fig:tri tile stripe attack Q} shows how to shift some guards in order to keep the invariant that all the vertices of the stripe have a guard.

\begin{figure}[ht]
    \centering
    
    \subfloat[\label{fig:tri tile stripe attack Q}]{
    \centering
    \begin{minipage}{0.24\textwidth}
    \centering
    \includegraphics[width=.8\textwidth]{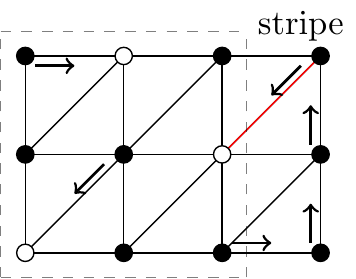}
    \end{minipage}
    \hspace{-20pt}
    \begin{minipage}{0.24\textwidth}
    \centering
    \includegraphics[width=.8\textwidth]{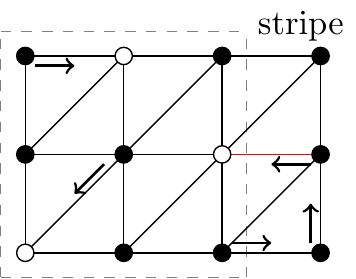}
    \end{minipage}
    }
    \subfloat[\label{fig:tri tile stripe attack Z}]{
    \centering
    \begin{minipage}{0.24\textwidth}
    \centering
    \includegraphics[width=.8\textwidth]{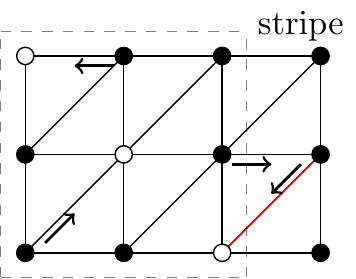}
    \end{minipage}
    \hspace{-20pt}
    \begin{minipage}{0.24\textwidth}
    \centering
    \includegraphics[width=.8\textwidth]{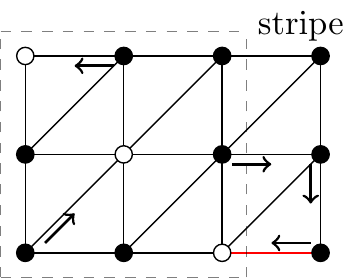}
    \end{minipage}
    }
    \caption{Shift when the attacked edge connects a tile with the stripe.
    a. Shift of guards from configuration of Figure~\ref{fig:tri grid tile Q} to the one of Figure~\ref{fig:tri grid tile Z}; b. Shift of guards from configuration of Figure~\ref{fig:tri grid tile Z} to the one of Figure~\ref{fig:tri grid tile Q}.}
    \label{fig:tri tile stripe attack}
\end{figure}

If guards are initially placed on vertices of $S_2$, the reasonings are identical but we exploit Figures~\ref{fig:tri grid 3x3 atk inside tile ZQ}, \ref{fig:tri grid tile pairs ZQ} and Figure~\ref{fig:tri tile stripe attack Z}.

Analogously to the case $h=2$, we deduce that both $S_1$ and $S_2$ are eternal vertex covers, and the correctness descends from shifting one set into the other and vice-versa.

Again, in order not to overburden the exposition, we choose not to mathematically detail the shift function but we simply refer to the figures.

It remains to evaluate $\rho^{\infty}(T_6(h,w))$.
$|S_1|=|S_2|= \frac{2}{3}(9w'h')+3h'w''+3w'h''+h''w''$ and the number of vertices of $T_6(h,w) =hw \geq 9 h'w'$.
Since $0 \leq h'',w'' \leq 2$, we have that
$\rho^{\infty}(T_6(h,w)) \leq 
\frac{2}{3}+\frac{2}{w}+\frac{2}{h} + \frac{4}{hw}$.
From the inequality $w \geq h$ we deduce that $\rho^{\infty}(T_6(h,w))$ is upper bounded by $\frac{2}{3}+\frac{4}{h}+\frac{4}{h^2}$.
\end{proof}

\subsection{Finite Octagonal Grids}
As in the previous cases, is not possible to exploit the method based on the Hamiltionan cycle.
Hence, we propose a defense strategy that requires guards on approximately $3/4$ of the vertices.
The idea of our strategy is to move only a subset of guards, placed on the vertices of two consecutive columns, leaving all the other guards fixed. 

\begin{theorem}
\label{th.finite_EVC_T8}
$\rho^{\infty}(T_8(h,w)) \le \frac{3}{4} + \frac{1}{2h} + \frac{1}{4h^2}$, $\forall\; w \geq h \geq 2$.
\end{theorem}
\begin{proof}
Let \(S\) the subset of vertices of $T_8(h,w)$ in which the guards are placed and let it include all vertices at odd $x$-coordinates and alternated vertices at even $x$-coordinates (not necessarily aligned at $y$-coordinates).
Note that on alternated columns we place $\lceil \frac{h}{2}\rceil$ guards and not $\lfloor \frac{h}{2} \rfloor$ in order to have enough guards during a defense.
Hence $|S|=\lfloor \frac{w}{2} \rfloor h + \lceil \frac{w}{2} \rceil \lceil \frac{h}{2} \rceil \leq \frac{3}{4} hw + \frac{h}{4} + \frac{w}{4} +\frac{1}{4}$, and in the hypothesis that $w\ge h$, $\rho^\infty (T_8(h,w)) \le \frac{3}{4} + \frac{1}{2h} + \frac{1}{4h^2}$.

Let us now consider the possible attacks and the strategies to apply.
The underlying idea of our strategy is to shift only a subset of guards, placed on vertices of two consecutive columns while leaving all the other guards fixed. As an invariant of the strategy, after every attack and defense, we will always have guards on every vertex at odd $x$-coordinate and at least on alternated vertices at even $x$-coordinates. This will guarantee the correctness of our strategy.

Let the attacked edge be \(\{(x',y'),(x'',y'')\}\), 
and let a guard lie on $(x',y')$; since there is no guard on $(x'', y'')$, $x''$ is even. 
W.l.o.g., we assume that $x''$ is either $x'$ or $x'+1$ and  $y''$ is either $y'$ or $y'+1$ (see Fig.~\ref{fig:octa_grid_attack}). 
In fact we omit the cases in which either $y''$ is equal to $y'-1$ or $x''$ is equal to $x'-1$, that can be handled symmetrically. 
Let us now focus on the subgraph of $T_8(h,w)$ induced by the vertices at $x$-coordinates $x'' -1$ and $x''$, and let the $y$-coordinate of the lower vertices of this subgraph be $0$. Refer to Fig.~\ref{fig:octa_grid_attack}.

\begin{figure}[ht]
\centering
\subfloat[\label{fig:octa_grid_attack obliquo}]{
\includegraphics[scale=0.8]{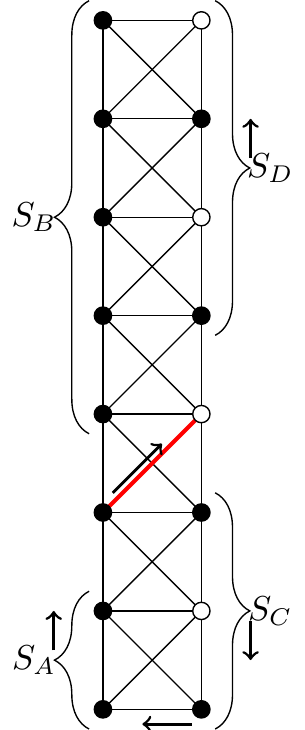}
}~
\subfloat[\label{fig:octa_grid_attack orizzontale}]{
\includegraphics[scale=0.8]{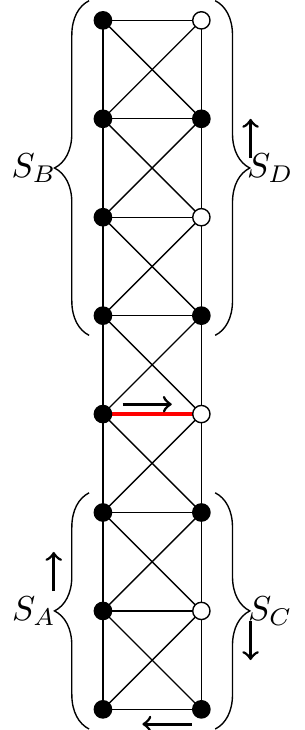}
}~
\subfloat[\label{fig:octa_grid_attack verticale 1}]{
\includegraphics[scale=0.8]{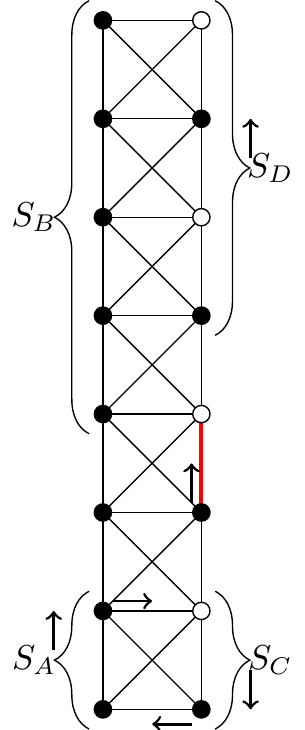}
}
\caption{
Possible cases in the proof of Theorem~\ref{th.finite_EVC_T8}.}
    \label{fig:octa_grid_attack}
\end{figure}

We can define the four (possibly empty) sets constituted by the vertices of this subgraph: \\
\(S_A = \{ (x''-1,0), \ldots, (x''-1, y'-1) \}\), \\
\(S_B = \{ (x''-1,y'+1), \ldots, (x''-1, h-1) \}\); \\ 
\(S_C = \{ (x'',0), \ldots, (x'', y'-1) \}\); if the attacked edge is of kind \(\{(x',y'),(x'+1,y'+1)\}\), then we add to set \(S_C\) also vertex \((x'',y')\); \\
\(S_D = \{ (x'',y''+1), \ldots, (x'', h-1) \}\).

Our strategy consists in shifting the guards on sets $S_A$, $S_C$, $S_D$ and few further guards.

Since at even $x$-coordinates there are $\lceil h/2 \rceil$ alternated guards, at least one vertex between $(x'',0)$ and $(x'', h-1)$ contains a guard; w.l.o.g.\ let it be $(x'',0)$ (otherwise we can rotate the layout of the graph w.r.t.\ its horizontal symmetry axis). 

Our defense strategy works as follows:
\begin{itemize}
    \item guard on $(x',y')$ moves to $(x'', y'')$ defending the attacked edge: $f((x',y'))=(x'', y'')$;
    \item all the guards lying on a vertex of $S_B$ remain where they are: $f((x''-1,i))=(x''-1,i)$ $\forall i=y'+1, \ldots, h-1$;
    \item all the guards lying on a vertex of $S_D$ move up: if there is a guard on \((x'',i)\) then $f((x'',i))=(x'', i+1)$ $\forall i=y''+1, \ldots, h-2$; 
    \item all the guards lying on a vertex of $S_C$ (except $(x'',0)$) shift down: if there is a guard on \((x'',i)\) then $f((x'',i))=(x'', i-1)$ $\forall i=1, \ldots, y'-1$; moreover, if the attacked edge is of kind $\{(x',y'),(x'+1,y'+1)\}$ then $f((x'',y'))=(x'',y'-1)$. Note that, at moment, both $(x'',0)$ and $(x'', 1)$ contain a guard, even if they are adjacent and have the same even $x$-coordinate;

    \item if the attacked edge is of kind $\{(x',y'),(x',y'+1)\}$, vertices $(x',y'-1)$ and $(x', y'-2)$ do not contain a guard even if they are adjacent; 
    thus guard on $(x'-1,y'-1)$ moves to its right: $f((x'-1,y'-1))=(x', y'-1)$;

    \item all the guards lying on a vertex of $S_A$ shift up:
    if there is a guard on \((x''-1,i)\) then $f((x''-1,i))=(x''-1, i+1)$ $\forall i=0, \ldots, y'-1$; note that at moment vertex $(x''-1,0)$ is not covered by any guard;
    \item guard on $(x'',0)$ shifts to its left: $f((x'',0))=(x''-1, 0)$.
    
    \item if the attacked edge is of kind $\{(x',y'), (x'-1, y'+1)\}$, we can horizontally flip $T_8(h,w)$ and get again an attacked edge of kind $\{(x',y'),(x'+1, y'+1)\}$.
\end{itemize}

Observe that, after these guards' shifts, column $x''-1$ still has guards on all its vertices, and column $x''$ still has guards at least on alternating vertices, so the invariant is kept.
\end{proof}

\section{Conclusions}

In this paper, we studied the vertex cover and eternal vertex cover numbers on regular grid graphs, both in the infinite and finite cases.
To this aim, we extended the concepts of minimum vertex cover and minimum eternal vertex cover sets to infinite grid graphs.

This study arose from observing that paths and square grids behave differently w.r.t.\ the eternal vertex cover number: this parameter is rather different for paths in the finite and infinite cases, while it is very similar for squared grids in the two cases.

We performed a systematic study for all regular grids, {\em i.e.} squared, hexagonal, triangular and octagonal grids.
The results are summarized in Tables~\ref{tab:infinite} and~\ref{tab:finite}.

The conclusion of this study is that having a not null second dimension helps for an effective defense strategy, indeed path remains the only graph for which $\rho^{\infty}$ is completely different when passing from the infinite to the finite cases.

\clearpage

\bibliographystyle{elsarticle-num}
\bibliography{references}
\end{document}